\documentclass[journal, draftcls, 12 pt, onecolumn, twoside]{IEEEtran} 
\usepackage{cite}

\ifCLASSINFOpdf
  \usepackage[pdftex]{graphicx}
\else
\fi
\usepackage[cmex10]{amsmath}
\usepackage{algorithmic}

\usepackage{array}

\usepackage{mdwmath}
\usepackage{mdwtab}

\usepackage{eqparbox}

\usepackage[tight,footnotesize]{subfigure}

\usepackage[caption=false,font=footnotesize]{subfig}
\usepackage{fixltx2e}

\usepackage{stfloats}
\usepackage{url}

\usepackage{amsfonts} 
\usepackage{epstopdf} 
\usepackage{amsmath,amssymb,color} 
\usepackage{mathrsfs} 
\usepackage{bbm} 
\usepackage{dsfont} 
\usepackage{acronym} 
\usepackage{enumerate} 

\newcommand{\noteStyleEmpty}[1]{\marginpar{}}

\newcommand{\mapNoteToDefinition}[2]{
\expandafter\renewcommand\csname#1\endcsname[1]{\csname#2\endcsname{##1}}%
}

\newcommand{\setnotelevel}[1]{%
\ifnum#1<1%
    \expandafter\mapNoteToDefinition{mynote}{noteStyleEmpty}%
    \expandafter\mapNoteToDefinition{mysubnote}{noteStyleEmpty}%
    \expandafter\mapNoteToDefinition{mysubsubnote}{noteStyleEmpty}%
\fi%
\ifnum#1=1%
    \expandafter\mapNoteToDefinition{mynote}{mynoteStyleDefault}%
    \expandafter\mapNoteToDefinition{mysubnote}{noteStyleEmpty}%
    \expandafter\mapNoteToDefinition{mysubsubnote}{noteStyleEmpty}%
\fi%
\ifnum#1=2%
    \expandafter\mapNoteToDefinition{mynote}{mynoteStyleDefault}%
    \expandafter\mapNoteToDefinition{mysubnote}{mysubnoteStyleDefault}%
    \expandafter\mapNoteToDefinition{mysubsubnote}{noteStyleEmpty}%
\fi%
\ifnum#1>2%
    \expandafter\mapNoteToDefinition{mynote}{mynoteStyleDefault}%
    \expandafter\mapNoteToDefinition{mysubnote}{mysubnoteStyleDefault}%
    \expandafter\mapNoteToDefinition{mysubsubnote}{mysubsubnoteStyleDefault}%
\fi%
}

\newtheorem{theorem}{Theorem}

\newtheorem{corollary}{Corollary}

\acrodef{RV}{random variable}

\acrodef{FT}{Fourier Transform}

\acrodef{i.i.d.}{independent, identically distributed}

\acrodef{p.m.}{probability measure}

\acrodef{p.d.f.}{probability density function}

\acrodef{c.d.f.}{cumulative distribution function}

\acrodef{ch.f.}{characteristic function}

\acrodef{AWGN}{additive white gaussian noise}

\acrodef{SNR}{signal-to-noise ratio}

\acrodef{LRT}{likelihood ratio test}

\acrodef{GLRT}{generalized likelihood ratio test}

\acrodef{LOS}{line-of-sight}

\acrodef{NLOS}{non-line-of-sight}

\acrodef{GDOP}{geometric dilution of precision}

\acrodef{GPS}{Global Positioning System}

\acrodef{FIM}{Fisher information matrix}

\acrodef{PEB}{position error bound}

\acrodef{WSN}{Wireless Sensor Network}

\acrodef{MAC}{medium access control}

\acrodef{RSS}{received signal strength}

\acrodef{RTT}{round-trip time}

\acrodef{MF}{matched filter}

\acrodef{ED}{energy detector}

\acrodef{ML}{maximum likelihood}

\acrodef{NL}{nonlinear}

\acrodef{MSE}{mean square error}

\acrodef{RMSE}{root mean square error}

\acrodef{ppm}{part-per-million}

\acrodef{ACK}{acknowledge}

\acrodef{UWB}{ultrawide bandwidth}

\acrodef{TNR}{threshold-to-noise ratio}

\acrodef{NLOS}{non line-of-sight}

\acrodef{LOS}{line-of-sight}

\acrodef{LS}{least squares}

\acrodef{IR-UWB}{impulse radio UWB}

\acrodef{FCC}{Federal Communications Commission}

\acrodef{TH}{time-hopping}

\acrodef{PPM}{pulse position modulation}

\acrodef{PAM}{pulse amplitude modulation}

\acrodef{MUI}{multi-user interference}

\acrodef{PDP}{power delay profile}

\acrodef{BPZF}{band-pass zonal filter}

\acrodef{SIR}{signal-to-interference ratio}

\acrodef{RFID}{radiofrequency identification}

\acrodef{WPAN}{wireless personal area networks}

\acrodef{WWLB}{Weiss-Weinstein lower bound}

\acrodef{DP}{direct path}

\acrodef{MF}{matched filter}

\acrodef{MMSE}{minimum-mean-square-error}

\acrodef{SBS}{serial backward search}

\acrodef{NBI}{narrowband interference}

\acrodef{WBI}{wideband interference}

\acrodef{INR}{interference-to-noise ratio}

\acrodef{CIR}{channel impulse response}

\acrodef{LRT}{likelihood ratio test}

\acrodef{RADAR}{RADAR}

\acrodef{MUR}{Multistatic RADAR}

\acrodef{e.m.}{electromagnetic}

\acrodef{CW}{continuous wave}

\acrodef{RF}{radiofrequency}

\acrodef{FCC}{Federal Communications Commission}

\acrodef{EIRP}{effective radiated isotropic power}

\acrodef{RCS}{radar cross section}

\acrodef{BAV}{balanced antipodal Vivaldi}

\acrodef{PRake}{partial Rake}

\acrodef{RTLS}{real time locating systems}

\acrodef{Hi-RADIAL}{High-accuracy RAdio Detection, Identification,
And Localization}

\acrodef{CRB}{Cram\'{e}r-Rao bound}

\acrodef{ZZB}{Ziv-Zakai bound}

\acrodef{TOA}{time-of-arrival}

\acrodef{TOF}{time-of-flight}

\acrodef{WSN}{wireless sensor network}

\acrodef{MAC}{medium access control}

\acrodef{RSS}{received signal strength}

\acrodef{TDOA}{time difference-of-arrival}

\acrodef{RF}{radiofrequency}

\acrodef{RTT}{round-trip time}

\acrodef{AOA}{angle-of-arrival}

\acrodef{MF}{matched filter}

\acrodef{ED}{energy detector}

\acrodef{ML}{maximum likelihood}

\acrodef{MUR}{Multistatic radar}

\acrodef{HDSA}{high-definition situation-aware}

\acrodef{RRC}{root raised cosine}

\acrodef{OFDM}{orthogonal frequency division multiplexing}

\acrodef{IF}{intermediate frequency}

\acrodef{PHY}{physical layer}

\acrodef{S-V}{Saleh-Valenzuela}

\acrodef{UHF}{ultra-high frequency}

\acrodef{PR}{pseudo-random}

\acrodef{SoC}{System on Chip}

\acrodef{SoP}{System on Package}

\acrodef{SPMF}{Single-Path Matched Filter}

\acrodef{IMF}{Ideal Matched Filter}

\acrodef{SCR}{signal-to-clutter ratio}

\acrodef{BEP}{bit error probability}

\acrodef{BER}{bit error rate}

\newtheorem{proposition}{Proposition}
\newtheorem{remark}{Remark}

\begin{document}

\newcommand{\ud}{\mathrm{d}}  
%
\title{Capacity Achieving Peak Power Limited Probability Measures: Sufficient Conditions for Finite Discreteness}
%
%
%

\author{Vincenzo~Zambianchi,~\IEEEmembership{Student Member,~IEEE,}
        Enrico~Paolini,~\IEEEmembership{Member,~IEEE,}
        and Davide~Dardari,~\IEEEmembership{Senior Member,~IEEE}
\thanks{The authors are with the Department
of Electrical, Electronics and Information Engineering ``G. Marconi'' (DEI), University of Bologna, via Venezia 52, Cesena (FC) 47521,
Italy. E-mail: \{vincenzo.zambianchi, davide.dardari, e.paolini\}@unibo.it. This work has been supported by the GRETA PRIN Project.}}

%
%

\markboth
    {submitted to IEEE Transactions on Information Theory}
    {Zambianchi et al.: Capacity Achieving Peak Power Limited Probability Measures}
%



\maketitle

\begin{abstract}
The problem of capacity achieving (optimal) input \ac{p.m.} has been widely investigated for several channel models with constrained inputs. So far, no outstanding generalizations have been derived. This paper does a forward step in this direction, by introducing a set of new requirements, for the class of real scalar conditional output \ac{p.m.}'s, under which the optimal input \ac{p.m.} is shown to be discrete with a finite number of probability mass points when peak power limited.
\end{abstract}


%
\IEEEpeerreviewmaketitle

\section{Introduction}
%
%
%
%
\IEEEPARstart{I}{n} recent years, a great interest has been rising in what can be called discrete input channel modeling. This theory takes its first steps from the study of classical (Gaussian) additive noise channels under input constraints. The class of channels with input limitations is important from a practical point of view since feasible systems do always have to deal with input constraints: Peak and average power are necessarily bounded. The first works in this field were the ones by Smith back in the 70's \cite{Smi:69, Smi:71}: He made forward steps with respect to Shannon's work \cite{Sha:48} considering an additive Gaussian noise channel in which the input is either peak or both peak and average power constrained. He discovered that, under both constraints, the capacity achieving input \ac{p.m.} is discrete with a finite number of probability mass points. This kind of \ac{p.m.}'s will be referred to as \emph{finitely discrete} throughout this paper. Smith's result was of notable importance since continuous inputs are not feasible in practice and have to be approximated with finitely discrete inputs. 

The finitely discrete feature was demonstrated to be the exact solution for the capacity achieving input \ac{p.m.} in the constrained additive scalar Gaussian noise channel model. This paved the way to several subsequent studies that, more recently, explored the finite discreteness of capacity achieving input \ac{p.m.}'s for other input constrained channel models, presenting quite disparate characteristics. Among them we cite \cite{ShaBar:95} and \cite{FayTroSha:01}, which inspired further works such as \cite{Tch:04} and \cite{ChaHraKsc:05}. Concerning the two last mentioned works, the former presents conditions on the \ac{p.m.} of an additive scalar channel noise, that are sufficient for the optimal bounded input \ac{p.m.} to have a finitely discrete support. The latter demonstrates that such a support is sparse (see \cite{ChaHraKsc:05} for definition) when the channel conditional output \ac{p.m.}, possibly not scalar, is Gaussian distributed. Subsequent works exploited the finitely discrete nature of the input \ac{p.m.} in some specific cases (e.g., \cite{FeiMat:07, LeiGeiWit:12}) but no further generalizations have been developed to the authors' knowledge.

In this paper, we consider a wide real scalar channel model and provide sufficient conditions on the conditional output \ac{p.m.} for the peak power limited capacity achieving input \ac{p.m.} to be finitely discrete. We establish this result without indicating any particular type of conditional output \ac{p.m.} nor any particular kind of the channel input-output law. Moreover, we prove that several peak power constrained additive channels as well as the peak power constrained Rayleigh fading channel fall in the developed framework as particular cases, whereas so far they have always been regarded as two distinct categories, necessitating different mathematical treatments. In this respect, the presented conditions extend the theory of peak power limited real input scalar channels.

The contribution is organized as follows. In Section~\ref{sec:Prel} all necessary notation and definitions are introduced, while in Section~\ref{sec:Channel} our main result is stated. This result is gradually proved in Sections~\ref{sec:Existence}, \ref{sec:Analyticity}, and \ref{sec:Input}. Some hints about uniqueness of the capacity achieving input \ac{p.m.} are provided in Section~\ref{sec:Unique}. The above mentioned examples are analysed in Section~\ref{sec:Examples}, while conclusions are drawn in Section~\ref{sec:conclusions}. Ancillary results necessary for the proof of the main theorem are deferred to Appendices~\ref{sec:PrelSteps}, \ref{sec:Der}, and \ref{sec:Theorems} while Appendix~\ref{sec:Qandq} provides some deeper explanations concerning the earlier discussed examples.

\section{Notation and Early Definitions}\label{sec:Prel}
In this section we present our notation and definitions coherently with the ones given by previous authors \cite{Smi:71, ChaHraKsc:05}.

Throughout this paper, $Y$ and $X$ represent the real scalar channel output and input \acp{RV}, respectively. We denote by $F(x)$ the input \ac{c.d.f.}, by $p_X(x)$ the input \ac{p.m.}, and by $p_{Y|X}(y|x)$ the conditional output \ac{p.m.} The input \ac{RV} $X$ is assumed to take values in the set $\mathbb{S}$, with $\mathscr{P}$ being the ensemble of possible \ac{p.m.}'s defined on that set. The corresponding class of \ac{c.d.f.}'s is denoted by $\mathscr{F}$. 
We have
\begin{align} \label{p_Y}
& p_Y(y)= 
\int_{\mathbb{S}} p_{Y|X}(y|x)p_X(x) \ud x = \int_{\mathbb{S}} p_{Y|X}(y|x)\ud F(x) = p_Y(y;p_X) = p_Y(y;F)
\end{align}
where we make explicit the dependence on $p_X(x)$ of the output \ac{p.m.} $p_Y(\cdot)$.\footnote{Here, and throughout the whole paper, one of the two equivalent formulations with $p_X(x)$ or $F(x)$ will be freely used as appropriately needed.}

Channel capacity is the supremum over the input \ac{p.m.} of the mutual information functional~\cite{CovTho:B91}
\begin{align} \label{MutInf}
\mathrm{I}(X;Y) & = 
\int_{\mathbb{R}} \int_{\mathbb{S}} p_{Y|X}(y|x) \log \frac{p_{Y|X}(y|x)}{p_Y(y;F)} \ud F(x) \ud y =  I(F)
\end{align}
where $\log(\cdot)$ denotes the base-2 logarithm.\footnote{In contrast, $\ln(\cdot)$ will denote the natural logarithm.}
Since only meaningless channel structure have zero capacity, we will assume channel capacity to be strictly positive and we will denote the capacity achieving (hence optimal) input \ac{p.m.} by $p_{X_0}(x)$. 
The mutual information functional can be further developed as
\begin{align} \label{I}
I(F) = H(F)-D(F)
\end{align}
where
\begin{align}
& 
H(F) \triangleq 
- \int_{\mathbb{R}} p_Y(y;F) \log p_Y(y;F) \ud y \nonumber
\end{align}
and
\begin{align}
& 
D(F) \triangleq 
- \int_{\mathbb{R}} \int_{\mathbb{S}} p_{Y|X}(y|x) \log p_{Y|X}(y|x) \ud F(x) \ud y. \nonumber
\end{align}
We can note how $D(F)$ depends in general on the input \ac{c.d.f.}, as opposed to what happens for an additive Gaussian channel (Smith's model, \cite{Smi:71}).

We also define the marginal information density and the marginal entropy density as
\begin{align}
i(x;F) \triangleq 
\int_{\mathbb{R}} p_{Y|X}(y|x) \log \frac{p_{Y|X}(y|x)}{p_Y(y;F)} \ud y \nonumber
\end{align}
and
\begin{align}
h(x;F) \triangleq 
- \int_{\mathbb{R}} p_{Y|X}(y|x) \log p_Y(y;F) \ud y \nonumber
\end{align}
respectively. These two densities are related as
\begin{align}
i(x;F)=h(x;F)-d(x) \nonumber
\end{align}
where 
\begin{align}
d(x)\triangleq - \int_{\mathbb{R}} p_{Y|X}(y|x) \log p_{Y|X}(y|x) \ud y. \nonumber
\end{align}
It is straightforward to show that the following three statements also hold:
\begin{align}\label{eq:IF}
I(F) = 
\int_{\mathbb{S}} i(x;F) \ud F(x)
\end{align}
\begin{align}\label{eq:HF}
H(F) = 
\int_{\mathbb{S}} h(x;F) \ud F(x)
\end{align}
and 
\begin{align}\label{eq:DF}
D(F) = 
\int_{\mathbb{S}} d(x) \ud F(x).
\end{align}
In this paper, \eqref{eq:IF}, \eqref{eq:HF}, and \eqref{eq:DF} are well-defined since $h(\cdot)$, $i(\cdot)$, and $d(\cdot)$ are finitely bounded under the conditions enunciated in Section \ref{sec:Channel}, as proven in Appendix \ref{sec:PrelSteps}.

\section{Framework Set Up and Main Result} \label{sec:Channel}
We consider a memoryless real scalar channel governed by a general input-output relationship in the form
\begin{align} \label{channel}
Y=\mathit{f}(X,\underline{\Theta})
\end{align}
where $X$ is the input \ac{RV} and $\underline{\Theta}$ a vector of nuisance parameters.
We do not impose further conditions on the input-output channel law $\mathit{f}(\cdot)$, which may be linear or nonlinear, additive in noise or multiplicative or both, with independent or correlated noises.

Throughout the paper, we consider a peak power constrained input \ac{RV} $X$ taking values in the bounded set (see Fig. \ref{Fig:S})
\begin{align}
\mathbb{S}=[-A,A] \cap \mathbb{A} \nonumber
\end{align}
where $[-A,A]$ is the compact real interval of radius $A$ 
and $\mathbb{A}$ represents an open subset of the complex extended input plane on which the conditional output \ac{p.m.} $p_{Y|X}(y|x)$ is analytic (hence continuous) in the input variable. 

\begin{figure}[!t]
    {\begin{center}\includegraphics[width=\columnwidth]{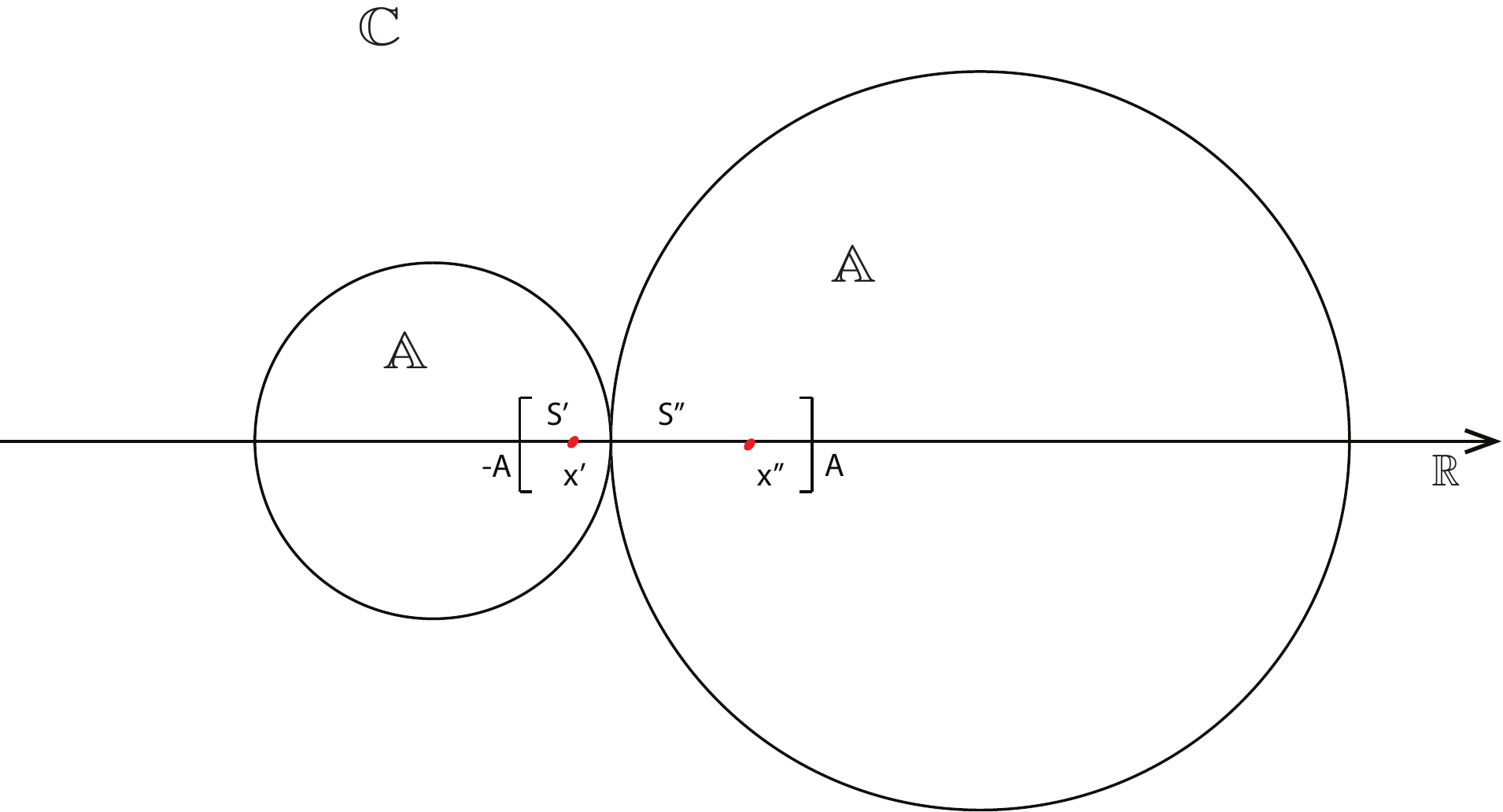}\end{center}}
\caption{Pictorial representation of the set $\mathbb S =[-A,A] \cap \mathbb{A}$ on which the input \ac{RV} takes its values.}
\label{Fig:S}
\end{figure}

The fundamental conditions on which our analysis relies may be summarized as follows:
\begin{enumerate}
\item \label{cond1}The conditional output \ac{p.m.} can be analytically extended to complex inputs, i.e., there exists an open set $\mathbb{A} \subseteq \mathbb C$ such that
\begin{align}
x \mapsto p_{Y|X}(y|x) \nonumber 
\end{align}
is an analytic map over $\mathbb{A}$, while 
\begin{align}
(x,y) \mapsto p_{Y|X}(y|x) \nonumber
\end{align}
is a continuous function over $\mathbb{A} \times \mathbb{R}$.
\item \label{cond2} There exist two functions $q(y)$ and $Q(y)$, both nonnegative, and bounded above, and integrable, such that  $\forall x \in \mathbb{S}$ we have
\begin{align} \label{qQ}
0 \le q(y)\le p_{Y|X}(y|x)\le Q(y) \le K < +\infty, \; \forall y \in \mathbb{R}
\end{align}
and the map
\begin{align}
y \mapsto Q(y)\log q(y) \nonumber
\end{align}
is integrable in $y$.
\item \label{cond3}The two integrals 
\begin{align}
&\int_{\mathbb{R}} p_{Y|X}(y|w)\log{p_{Y|X}(y|w)} \ud y \nonumber\\
&\int_{\mathbb{R}} p_{Y|X}(y|w)\log{p_{Y}(y;p_X)} \ud y \nonumber
\end{align}
are uniformly convergent (see \cite{Lan:B99} for definition) $\forall w \in \mathbb{D}$, for some $\mathbb{D}$ such that $\mathbb{S}\subset \mathbb{D} \subseteq \mathbb{A}$.\footnote{For the sake of clarity, here and elsewhere in the paper a generic input value is denoted by $x$ or $w$ whenever the input is considered strictly real or complex extended, respectively.} 
\item \label{cond4} For each of the maximally extended connected regions forming $\mathbb{S}$ (we call them $S',S'', \ldots$), one of the following three conditions holds:
\begin{enumerate}
\item \label{opt1}
there exist $x', x'', \ldots \in S', S'', \ldots$ (see Fig. \ref{Fig:S}) and corresponding \ac{c.d.f.}'s $F', F'', \ldots$ with
\begin{align} \label{x'}
\log p_{Y|X}(y|x') - \log q(y) < I(F'), \; \forall y \in \mathbb{R} 
\end{align}
and analogously for the other regions, where $I(F')$ is the mutual information between the output and input variable when the input is distributed according to $F'(x)$.
\item \label{opt2}
for all real input \ac{p.m.}'s $p_X(x)$, there exist $x', x'', \ldots \in S', S'', \ldots$ (see Fig. \ref{Fig:S}) such that $p_{Y|X}(y|x')$ is the unique conditional output \ac{p.m.} satisfying
\begin{align}
\min_{x\in S'} D_{KL}(p_{Y|X}(y|x)||p_Y(y;p_X))=D_{KL}(p_{Y|X}(y|x')||p_Y(y;p_X)) \nonumber
\end{align}
and analogously for the other regions, where $D_{KL}$ denotes the Kullback-Leibler divergence.
\item \label{opt3}
for all real input \ac{p.m.}'s $p_X(x)$, there exist pairs of distinct points $(x_1',x_2'), (x_1'',x_2''), \ldots$ $\in S'\times S', S''\times S'', \ldots$ such that
\begin{align}
D_{KL}(p_{Y|X}(y|x_1')||p_Y(y;p_X)) \ne D_{KL}(p_{Y|X}(y|x_2')||p_Y(y;p_X)) \nonumber
\end{align}
and analogously for the other regions.
\end{enumerate}
\end{enumerate}

\begin{remark}
The here stated conditions do not impose any peculiar kind of conditional output \ac{p.m.}, as it was the case in \cite{Smi:69, Smi:71, ChaHraKsc:05}, nor any particular channel law, as it was done in \cite{Tch:04}. We also underline that the input set compactness, deeply exploited in \cite{ChaHraKsc:05}, is \emph{not} a required condition here. Examples, considered in Section~\ref{sec:Examples}, further show the presented theory to extend the previously known treatments.
\end{remark}
\bigskip
We are now in a position to state the main result of this contribution.
\begin{theorem} \label{Th1}
Every real scalar and peak power constrained input channel, whose conditional output \ac{p.m.} fulfils the aforementioned conditions 1 to 4, has a finitely discrete capacity achieving input \ac{p.m.}
\end{theorem}
\medskip
The remainder of this paper is devoted to prove Theorem \ref{Th1}. The proof requires some intermediate steps: In particular, Section~\ref{sec:Existence} proves that the capacity achieving input \ac{p.m.} exists and also states, as a corollary, Kuhn-Tucker's conditions on the marginal information density (defined in Section~\ref{sec:Prel}) for an input \ac{p.m.} to be optimal. Section~\ref{sec:Analyticity} proves the analyticity of the marginal information density which is exploited in Section~\ref{sec:Input}, alongside the corollary statement, to finally prove the finitely discrete nature of the capacity achieving input \ac{p.m.} support. Besides, Section~\ref{sec:Unique} hints in the direction of proving uniqueness of the optimal input \ac{p.m.}\footnote{Uniqueness was not proved in general neither in \cite{Tch:04} nor in \cite{ChaHraKsc:05}.}

\section{Existence of a Capacity Achieving Input \ac{p.m.}} \label{sec:Existence}
Following the approach in \cite{Smi:71, Smi:69}, in this section we demonstrate that an \emph{optimal} input \ac{p.m.} exists and that Kuhn-Tucker's conditions are necessary and sufficient for optimality. Some basic results in optimization theory are first reviewed \cite{Smi:69, Smi:71, Lue:B69}.

A map $f: \Omega \mapsto \mathbb{R}$, where $\Omega$ is a convex space, is said to be weakly differentiable in $\Omega$ if, for $\theta \in [0,1]$ and $x_0 \in \Omega$, the map $f'_{x_0}:\Omega \rightarrow \mathbb{R}$, defined as
\begin{align}
f'_{x_0}(x) = \lim\limits_{\theta \rightarrow 0} \frac{f[(1-\theta)x_0 + \theta x] - f(x_0)}{\theta} \nonumber
\end{align}
exists for all $x$ and $x_0$ in $\Omega$.
Besides, $f$ is said to be
concave if, for all $\theta \in [0,1]$ and for all $x$ and $x_0$ in $\Omega$,
\begin{align}
f[(1-\theta)x_0 + \theta x] \ge (1-\theta)f(x_0) + \theta f(x). \nonumber
\end{align}

\medskip
\begin{theorem}[Optimization Theorem \cite{Lue:B69}]\label{Th:Op}
Let $f$ be a continuous, weakly differentiable, and
concave map from a compact, convex topological space $\Omega$ to $\mathbb{R}$, and define
\begin{align}
C \triangleq \sup_{x \in \Omega} f(x)\, . \nonumber
\end{align}
Then:
\begin{enumerate}
\item $C = \max_{x \in \Omega} f(x) = f(x_0)$ for some $x_0 \in \Omega$;
\item $f(x_0)=C$ if and only if $f'_{x_0}(x) \le 0 \;\, \forall x \in \Omega$.
\end{enumerate}
\end{theorem}

\medskip
Exploiting the above results from optimization theory, we have the following proposition.

\medskip
\begin{proposition}\label{prop1}
Let $I(F)$ be the mutual information functional between $X$ and $Y$, as defined in \eqref{MutInf}. Then, under an input peak power constraint and conditions 1 and 2 of Sec. \ref{sec:Channel}, there exists an $F_0 \in \mathscr{F}$ (equivalently a $p_{X_0} \in \mathscr{P}$) such that
\begin{align}
C = I(F_0) = \max_{F \in \mathscr{F}} I(F). \nonumber
\end{align}
Moreover, a necessary and sufficient condition for the input \ac{c.d.f.} $F_0$ to maximize $I(F)$, i.e.,
to achieve capacity, is
\begin{align}\label{WeakDer}
\int\limits_{\mathbb{S}} i(x;F_0)\ud F(x) \le I(F_0), \quad \forall F \in \mathscr{F}.
\end{align}
\end{proposition}
\noindent 
\begin{proof} As from Theorem~\ref{Th:Op}, it suffices to show that $\mathscr{F}$ 
is convex and compact in some topology and that $I:\mathscr{F} \mapsto \mathbb{R}$ 
is continuous, concave and weakly differentiable. The necessary and sufficient condition \eqref{WeakDer} also follows from Theorem \ref{Th:Op}, as it will be shown.

\paragraph{Convexity and Compactness} The convexity of $\mathscr{F}$, i.e. the fact that
\begin{align}
F_{\theta}(x)=(1-\theta)F_{1}(x)+\theta F_{2}(x) \nonumber
\end{align} 
still belongs to $\mathscr{F}$ for each $F_{1}$, $F_{2}$ in $\mathscr{F}$ and for each $\theta \in [0,1]$, 
is immediate.
The compactness of $\mathscr{F}$ in the L\`evy metric\footnote{The corresponding distance is here indicated with $d(\cdot,\cdot)$.} topology (as defined in \cite{Smi:69}) follows from Helly's Weak Compactness Theorem (see Appendix \ref{sec:Theorems}) and from the fact that convergence in the L\`evy metric is equivalent to complete convergence \cite{Mor:B68}, which on a bounded interval 
is equivalent to weak convergence. 

\paragraph{Continuity} The continuity of functional $I(F)$ 
descends from the Helly-Bray Theorem (see Appendix \ref{sec:Theorems}), 
according to which $d(F_n,F) \xrightarrow[n]{} 0$ implies $I(F_n) \xrightarrow[n]{} I(F)$, 
provided the boundedness and continuity in $x$ of $i(x;F)$. The latter two properties are demonstrated in Appendix~\ref{sec:PrelSteps} (continuity of $i(x;F)$ is a consequence also of analyticity discussed in Section~\ref{sec:Analyticity}).

\paragraph{Concavity} For what concerns $I(F)$ 
being concave, we can note how
\begin{align}
p_Y(y;F_\theta) & = p_Y(y;(1-\theta)F_1 + \theta F_2) \nonumber\\
& = \int_{\mathbb{S}} p_{Y|X}(y|x) [(1-\theta)\ud F_1(x) + \theta \ud F_2(x)] = (1-\theta) p_Y(y;F_1) + \theta p_Y(y;F_2) \nonumber
\end{align}
and 
\begin{align}\label{Dtheta}
D \left( (1-\theta) F_1 + \theta F_2 \right) &= 
- \!\!\int_{\mathbb{R}}\! \int_{\mathbb{S}}\!\! p(y|x) \log p(y|x) [(1-\theta) \ud F_1(x) + \theta \ud F_2(x)] \ud y \nonumber \\
&= (1-\theta)D(F_1) + \theta D(F_2).
\end{align}
Hence, we have that
\begin{align}
I((1-\theta) F_1 + \theta F_2) \ge (1-\theta)I(F_1) + \theta I(F_2)\nonumber
\end{align}
is equivalent, from \eqref{I} and \eqref{Dtheta}, to
\begin{align}\label{H}
H((1-\theta) F_1 + \theta F_2) \ge (1-\theta) H(F_1) + \theta H(F_2).
\end{align}
Inequality \eqref{H} may be proved as follows:
\begin{align}
H((1-\theta) F_1 + \theta F_2) &= 
- \int_{\mathbb{R}}\! p_Y(y;(1\!-\!\theta) F_1\! + \!\theta F_2) \log p_Y(y;F_\theta) \ud y \nonumber \\
& = - \int_{\mathbb{R}} \left[ (1-\theta) p_Y(y;F_1) + \theta p_Y(y;F_2) \right] \log p_Y(y;F_\theta) \ud y \nonumber \\
&  \overset{\mathrm{(a)}}{\ge} -(1-\theta) \!\int_{\mathbb{R}} p_Y(y;F_1) \log p_Y(y;F_1) \ud y 
- \theta \!\int_{\mathbb{R}} p_Y(y;F_2) \log p_Y(y;F_2) \ud y \nonumber \\
& = (1-\theta) H(F_1) + \theta H(F_2) \nonumber
\end{align}
where $\mathrm{(a)}$ exploits Gibbs' inequality \cite{CovTho:B91}, which states that for any two random variables, $Z_1$ and $Z_2$, we have
\begin{align}
-\int_{\mathbb{R}} p_{Z_1}(z) \log p_{Z_1}(z) \ud z \le -\int_{\mathbb{R}} p_{Z_1}(z) \log p_{Z_2}(z) \ud z \nonumber
\end{align}
with equality if and only if 
\begin{align}
p_{Z_1}(z)= p_{Z_2}(z). \nonumber
\end{align}
Hence, concavity of $I(\cdot)$ is proven and equality holds if and only if $p_Y(y;F_1)= 
p_Y(y;F_2)$. 

\paragraph{Weak Differentiability}As proven in Appendix~\ref{sec:Der}, for arbitrary $F_1$ and $F_2$  in $\mathscr{F}$  we have
\begin{align} \label{Derivative}
& \phantom{=\;} \lim_{\theta \to 0} \frac{ I((1-\theta)F_1 + \theta F_2) - I(F_1) } {\theta} 
= \int_{\mathbb{S}} i(x;F_1) \ud F_2(x) - I(F_1).
\end{align}
The proof of weak differentiability is completed by observing that $i(x;F)$ is finitely bounded (Appendix \ref{sec:PrelSteps}), which guarantees the existence of the integral in the right-hand side of \eqref{Derivative}.

\bigskip
Since all hypotheses of Theorem~\ref{Th:Op} are satisfied, the optimal input \ac{p.m.} exists in $\mathscr{P}$. 
Furthermore, from \eqref{Derivative}, it is immediate to derive the necessary and sufficient condition \eqref{WeakDer}. 
\end{proof}

\medskip
The following corollary of Proposition \ref{prop1} states the Kuhn-Tucker's conditions that will be used in Section~\ref{sec:Input} to prove the final result.

\medskip
\begin{corollary}[Kuhn-Tucker's Conditions] \label{cor} \label{subsec:Corollary}
Let $p_{X_0}$ be an arbitrary \ac{p.m.} in $\mathscr{P}$. Let $S_0$ denote the set of mass points of $p_{X_0}$ on $\mathbb{S}$.\footnote{The set $S_0$ is defined independently of the discreteness or continuity of the input \ac{p.m.}} Then $p_{X_0}$ is optimal if and only if
\begin{align}
\begin{cases}
& i(x;p_{X_0}) \le I(p_{X_0}), \quad \forall x \; \in \mathbb{S}\\ 
& i(x;p_{X_0}) = I(p_{X_0}), \quad \forall x \; \in S_0 \nonumber
\end{cases}
\end{align}
\end{corollary}
\begin{proof} Even if Proposition \ref{prop1} requires a different demonstration, the here stated corollary can be proved in the same way as done in \cite{Smi:69,Smi:71}.
\end{proof}

\section{Analyticity of $i(w;p_X)$} \label{sec:Analyticity}
In this section we prove that $i(x;p_X)$ can be analytically extended to $i(w;p_X)$, $\forall w \in \mathbb{D}$. This step is necessary as a starting point for the capacity achieving input \ac{p.m.} characterization in Section~\ref{sec:Input}.

First, we extend $i(x;p_X)$ to the analyticity region $\mathbb{A}$ of $x\mapsto p_{Y|X}(y|x)$ as
\begin{align}
i(w;p_X)\triangleq \int\limits_{\mathbb{R}} p_{Y|X}(y|w) \log \frac{p_{Y|X}(y|w)}{p_Y(y;p_X)} \ud y \nonumber
\end{align}
$\forall w \in \mathbb{A}$ where convergence holds.\footnote{Convergence is guaranteed inside $\mathbb{S}$, as proven in Appendix \ref{sec:PrelSteps}.} We now apply the Differentiation Lemma (see Appendix \ref{sec:Theorems}, with $I=\mathbb{R}$, $U=\mathbb{D}$), to the functions
\begin{align}
&f_1(w,y)=p_{Y|X}(y|w) \log p_{Y|X}(y|w), \nonumber\\
&f_2(w,y)=p_{Y|X}(y|w) \log p_{Y}(y;p_X). \nonumber
\end{align}
The two functions are continuous (see Section~\ref{sec:Channel}) over $\mathbb{D}\times \mathbb{R}$.\footnote{$\mathbb{D}$ has to exclude the possibility for $p_{Y|X}(y|w)$ to be real negative valued, this to ensure continuity of the principal value complex logarithm.} Moreover, from conditions in Section~\ref{sec:Channel}, they are uniformly integrable over $\mathbb{R}$ and, 
being compositions of analytic functions, they are analytic. 
The difference of the two analytic (from Differentiation Lemma) integral functions
\begin{align}
\int\limits_{\mathbb{R}} f_1(y,w) \ud y - \int\limits_{\mathbb{R}} f_2(y,w) \ud y \nonumber
\end{align}
is analytic 
on $\mathbb{D}$. This means that $i(w;p_X)$ is an analytic function over $\mathbb{D}$.

\section{Capacity Achieving Input \ac{p.m.} Characterization} \label{sec:Input}
In this section we finally prove the finite discreteness of the capacity achieving input \ac{p.m.}.\newline
Define $v(w)$ as\footnote{Recall that a generic input value is denoted by $x$ and $w$ when the input is considered strictly real or complex extended, respectively.}
\begin{align}\label{v(w)}
v(w) &\triangleq \int\limits_{\mathbb{R}} p_{Y|X}(y|w) \left[-\log{\left( \frac{p_Y(y;p_{X_0})}{p_{Y|X}(y|w)} \right)} - I(p_{X_0})\right]\ud y \nonumber\\
& = i(w;p_{X_0})-I(p_{X_0})
\end{align}
where $p_{X_0}(x)$ is a capacity achieving input \ac{p.m.} Recall from Section~\ref{sec:Channel} that $S', S'', \ldots$ are the maximally extended connected regions forming $\mathbb{S}$, while $S_0', S_0'', \ldots$ is the corresponding decomposition for $S_0$ (the support of $p_{X_0}(x)$), i.e., $S_0'$ is the set of points of $S_0$ in $S'$, $S_0''$ is the set of points of $S_0$ in $S''$, and so on. Note that, if each of the optimal input domain decomposition sets were not finitely discrete, then, for the Bolzano-Weierstrass Theorem, it would have an accumulation point in the corresponding connected subregion of $\mathbb{S}$ and thus, by the identity principle of analytic functions and Corollary \ref{cor}, $v(w)=0$ in that entire subregion. From \eqref{v(w)}, $v(w)=0$ means
\begin{align}
-\int\limits_{\mathbb{R}} p_{Y|X}(y|w) \log{\left( \frac{p_Y(y;p_{X_0})}{p_{Y|X}(y|w)} \right)} \ud y -I(p_{X_0})=0. \nonumber
\end{align}
In the following, for notation convenience, suppose to consider the $S'$ subregion of $\mathbb{S}$.

In case one of the first two options \ref{opt1}, \ref{opt2} presented in Section~\ref{sec:Channel} is verified and since $v(w)=0$ on the entire considered subregion, we must have:
\begin{align}
v(x') \! = \! -\int\limits_{\mathbb{R}} p_{Y|X}(y|x') \log{\left( \frac{p_Y(y;p_{X_0})}{p_{Y|X}(y|x')} \right)} \ud y -I(p_{X_0})=0 \nonumber
\end{align}
also for the corresponding particular value $x'$, whose existence was supposed in Section~\ref{sec:Channel}. However this is in clear contradiction with either
\begin{align}
v(x') & =  \int\limits_{\mathbb{R}} p_{Y|X}(y|x') \log{\left( \frac{p_{Y|X}(y|x')}{p_Y(y;p_{X_0})} \right)} \ud y -I(p_{X_0}) \nonumber\\
&\le \int\limits_{\mathbb{R}} p_{Y|X}(y|x') \underbrace{\bigg( \log{p_{Y|X}(y|x')} - \log{q(y)} \bigg)}_{< I(F') \; \mathrm{see \; eq.}\eqref{x'}} \ud y -I(p_{X_0}) < I(F')-I(F_0) \le 0. \nonumber
\end{align}
or
\begin{align}
v(x') & =  \int\limits_{\mathbb{R}} p_{Y|X}(y|x') \log{\left( \frac{p_{Y|X}(y|x')}{p_Y(y;p_{X_0})} \right)} \ud y -I(p_{X_0}) \nonumber\\
&= D_{KL}(p_{Y|X}(y|x')||p_Y(y;p_{X_0})) - I(p_{X_0}) < D_{KL}(p_{Y|X}(y|x)||p_Y(y;p_{X_0}))-I(p_{X_0}) = 0. \nonumber
\end{align}
If vice versa the third option \ref{opt3} holds, it follows
\begin{align}
v(x_1')= D_{KL}(p_{Y|X}(y|x_1')||p_Y(y;p_{X_0})) - I(p_{X_0}) \ne D_{KL}(p_{Y|X}(y|x_2')||p_Y(y;p_{X_0})) - I(p_{X_0}) = 0 \nonumber
\end{align}
and again a contradiction occurs.

This finally proves that the hypothesis to have an infinite set of mass points $S_0$ was wrong, hence the input \ac{RV} $X$ can take only on a finitely discrete set of values.

\section{About Uniqueness}\label{sec:Unique}
The so far developed conditions on the capacity achieving input \ac{p.m.} do not guarantee also its uniqueness. In this direction, a further property that all eventual optimal input \ac{p.m.}'s must satisfy with respect to any other capacity achieving \ac{p.m.} can be outlined.\newline Consider all the optimal input \ac{p.m.}'s\footnote{In the previous sections, we proved that they belong to $\mathscr{P'}$, the restriction of $\mathscr{P}$ to the class of finitely discrete generalized functions defined on a finite number of probability mass points in the input support $\mathbb{S}$.} and denote the $i$-th of them by $p_{X_i}(x)$. Then, the following proposition holds.
\begin{proposition}\label{prop2}
All the optimal input \ac{p.m.}'s of a channel model satisfying conditions 1-4 in Section~\ref{sec:Channel}, must fulfil the condition
\begin{align}
i(x;p_{X_0})=I(p_{X_0}), \quad \forall x \in S_i \nonumber
\end{align}
$S_i \subset \mathbb{S}$ being the support of $p_{X_i}(x)$. 
\end{proposition}
\begin{proof}
Let $p_{X_0}(x)$ and $p_{X_1}(x)$ be two optimal input \ac{p.m.}'s (whose existence is guaranteed by Proposition \ref{prop1}), both with a finitely discrete support. Then also $(1-\theta)p_{X_1}(x)+\theta p_{X_0}(x)$ is capacity achieving, since the mutual information functional is concave (see Theorem \ref{ConFunc} in Appendix \ref{sec:Theorems}). This fact yields the weak derivative $I'_{p_{X_0}}(p_{X_1})$ to be null. Recall the probability mass points in $S_0$ and $S_1$ $x_m$ and $x_n$, and the correspondent probability $b_m$ and $a_n$, respectively. In addition suppose that the condition enunciated in Proposition 2 is not verified, i.e., $i(x;p_{X_0})<I(p_{X_0})$ for at least one of the $x_n \in S_1$, where the order relation is imposed by Corollary \ref{cor}. The cited weak derivative expression becomes
\begin{align}
\int\limits_{\mathbb{S}} i(x;p_{X_0})\left[p_{X_1}(x)-p_{X_0}(x)\right] \ud x &= \sum\limits_{n} a_n i(x_n;p_{X_0}) - \sum\limits_{m} b_m i(x_m;p_{X_0}) \nonumber\\
&< I(p_{X_0})\sum\limits_{n} a_n - I(p_{X_0})\sum\limits_{m} b_m = 0. \nonumber
\end{align}
A contradiction has arisen since $I'_{p_{X_0}}(p_{X_1})=0$ and $I'_{p_{X_0}}(p_{X_1})<0$, which completes the proof.
\end{proof}
This Proposition \ref{prop2} does not provide uniqueness of the capacity achieving input \ac{p.m.}, nevertheless it tightens the conditions for an input \ac{p.m.} to be optimal. Future attempts will be made aiming to prove uniqueness.

\section{Examples}\label{sec:Examples}
This section is divided in two subsections. The first one proves that any peak power constrained channel with additive noise satisfies condition \ref{cond4}, stated in Section~\ref{sec:Channel} and, therefore, it belongs to the general class of channels treated in this paper upon fulfilling also conditions \ref{cond1}, \ref{cond2}, and \ref{cond3}.\footnote{The fulfilment of conditions 1-3 must be checked case by case, but it is expected to be a simple verification.} The second one proves that the Rayleigh fading channel undergoes all the conditions in Section~\ref{sec:Channel}. With respect to the theory proposed in \cite{Tch:04}, we underline that the conditions in Section~\ref{sec:Channel} are less stringent, so a wider set of additive channels is characterized. 
\subsection{Additive Channels}\label{AddChannels}
Consider an additive channel model $Y=X+N$, where $N$ is the noise \ac{RV}. The marginal information density can be rewritten as
\begin{align}
i(x;p_X)&=\int_{\mathbb{R}}p_N (y-x)\log p_{N} (y-x)\ud y - \int_{\mathbb{R}}p_{N}(y-x)\log p_{Y}(y;p_X)\ud y \nonumber\\
&= k - \int_{\mathbb{R}}p_{N}(y-x)\log p_{Y}(y;p_X)\ud y \nonumber
\end{align}
where $k$ is constant as it can be easily shown with an ordinary variable substitution. The second term is in the form of convolution and admits \ac{FT} since $p_N(\cdot)$ is integrable on $\mathbb{R}$ and $\log p_Y(y;p_X)=u(y)$ is locally integrable hence transformable at least in the sense of distributions. Now assume the marginal information density is equal to a constant $c_1$: Its \ac{FT} would then be
\begin{align}
\Psi_N(2\pi f) U(f)=c_1\delta (f) \nonumber
\end{align}
where $\Psi_N(\cdot)$ denotes the characteristic function of the \ac{RV} $N$, defined as
\begin{align}
\Psi_N(f)=\mathbb{E}[\exp \{jxf\}]=\int_{\mathbb{R}} p_N(x)\exp \{jxf\} \ud x \nonumber.
\end{align}
The only case for this to hold is $u(y)$ being a constant itself: This is however contradictory since $u(y)=c_2$ implies $p_Y(y;p_X)=2^{c_2}$, which is clearly an absurd, and hence condition \ref{opt3} stands.
\subsection{Rayleigh Fading Channel}\label{RayChannel}
Consider the Rayleigh fading channel conditional output \ac{p.m.}, as defined in \cite{FayTroSha:01},
\begin{align}
p_{Y|X}(y|x)&=\frac{1}{1+x^2} \exp \left\{-\frac{y}{1+x^2}\right\}\nonumber \\
&=s\exp \{-ys\} \nonumber
\end{align}
and assume the channel input $X$ is subject to a peak power constraint $A$ as defined in Section~\ref{sec:Channel}. Since this conditional \ac{p.m.} derives from normalizations of the original input and output modules, $U$ and $V$ in \cite{FayTroSha:01}, this is a real scalar memoryless channel whose output takes values in $[0,+\infty)$. \newline
We now assess that the four conditions stated in Section~\ref{sec:Channel} are fulfilled.
\begin{enumerate}[\indent(a)]
\item
It is immediate to verify that condition \ref{cond1} holds over the set $\mathbb{A}=\mathbb{C} \smallsetminus \{-j,j\}$. 
\item
Concerning condition \ref{cond2}, let us define
\begin{align}
Q(y)=\begin{cases}
1, & \quad 0 \le y \le c(A^2+1) \\
\frac{1}{y^{1+\gamma}}, &\quad y>c(A^2+1) \nonumber
\end{cases}
\end{align}
where parameter $\gamma$ fulfils $\gamma <1$ and $c$ is a constant such that $c>2$ (the details are provided in Appendix~\ref{sec:Qandq}). Moreover, let us define
\begin{align}
q(y)=\begin{cases}
\frac{1}{1+A^2}\exp \left\{-\frac{y}{1+A^2}\right\}, & \quad 0 \le y \le \frac{(1+A^2)\ln(1+A^2)}{A^2} \\
\exp \{-y\}, & \quad y>\frac{(1+A^2)\ln(1+A^2)}{A^2} \nonumber
\end{cases}
\end{align}
where $y_2=\frac{(1+A^2)\ln(1+A^2)}{A^2}$ is the solution of $\frac{1}{1+A^2}\exp \left\{-\frac{y}{1+A^2}\right\} = \exp \{-y\}$. The two functions $q(\cdot)$ and $Q(\cdot)$ satisfy inequality \eqref{qQ}, as rigorously proven in Appendix~\ref{sec:Qandq}. Furthermore, both of them are nonnegative,  superiorly bounded, and integrable over the output domain $[0,+\infty)$. Besides $Q(y)\log q(y)$ is integrable over $[0,+\infty)$, which may be shown by analysing integrability over the tail.\footnote{$Q(y)\log q(y)$ is locally integrable since it is continuous.} We have
\begin{align}
\int_{y_3}^{+\infty} Q(y)\log q(y)\ud y &= \int_{y_3}^{+\infty} \frac{1}{y^{1+\gamma}} \exp \{-y\} \ud y \nonumber\\
&= \left[ -\frac{y^{-\gamma}}{\gamma} \exp \{-y\} \right]_{y_3}^{+\infty} - \int_{y_3}^{+\infty} \frac{y^{-\gamma}}{\gamma} \exp \{-y\} \ud y \nonumber
\end{align} 
which is finite. The considered $y_3$ is sufficiently large to guarantee that the expressions employed for $Q(y)$ and $q(y)$ are the proper ones. 
\item
We now consider condition \ref{cond3}. The integral 
\begin{align}
\int_{0}^{+\infty} p_{Y|X}(y|w) \log p_Y(y;p_X) \ud y \nonumber
\end{align} 
is uniformly convergent on $\mathbb{D}=\left\{w : \Re \{\frac{1}{1+w^2}\} \geq a_1, |\frac{1}{1+w^2}| \le a_2 \right\}$, with strictly positive $a_1$ and $a_2$, and with $a_1$ ensuring that $\mathbb{S} \subset \mathbb{D}$. Uniform convergence holds since, for each $w \in \mathbb{D}$, given $\epsilon$, there exist $B_0<B_1<B_2$ such that
\begin{align}
\big| \int_{B_1}^{B_2} p_{Y|X}(y|w) \log p_Y(y;p_X) \ud y \big| & \le \int_{B_1}^{B_2} \left| p_{Y|X}(y|w) \log p_Y(y;p_X) \right| \ud y \nonumber\\
& = \int_{B_1}^{B_2} \left| \frac{1}{1+w^2} \right| \left| \exp \left\{ -\frac{y}{1+w^2} \right\} \log p_Y(y;p_X) \right| \ud y \nonumber \\
& = \int_{B_1}^{B_2} \left| \frac{1}{1+w^2} \right| \left| \exp \left\{ -y\Re\left\{\frac{1}{1+w^2}\right\} \right\} \log p_Y(y;p_X) \right| \ud y \nonumber \\
& \le -\int_{B_1}^{B_2} \frac{1}{y^{3}} \log q(y) \ud y <\epsilon \nonumber
\end{align}
as $\left| \frac{1}{1+w^2} \exp \left\{ -y\Re\left\{\frac{1}{1+w^2}\right\} \right\} \right|$ is minor in a definitive manner in $y$ than $1/y^3$ regardless of $w \in \mathbb{D}$.\footnote{This is guaranteed by the existence of a maximum for $|\frac{1}{1+w^2}|$ and a non zero minimum for $\Re\left\{\frac{1}{1+w^2}\right\}$ on $\mathbb{D}$.} To prove the result it is also necessary to employ \eqref{f_g} in Appendix~\ref{sec:PrelSteps} and to choose $B_0$ in such a way that $B_0>y_2$ and $\frac{1}{B_0}<\epsilon$. The choice for $\mathbb{D}$ is dictated by the necessity to guarantee the existence of a uniform upper bound for $| p_{Y|X}(y|w) \log p_Y(y;p_X) |$. Analogously, also 
\begin{align}
\int_{0}^{+\infty} p_{Y|X}(y|w) \log p_{Y|X}(y|w) \ud y \nonumber
\end{align} 
is uniformly convergent on $\mathbb{D}$. In fact, for each $w \in \mathbb{D}$, given $\epsilon$, there exist $B_0<B_1<B_2$ such that
\begin{align}
\big| \int_{B_1}^{B_2} p_{Y|X}(y|w) & \log p_{Y|X}(y|w) \ud y \big| \le \int_{B_1}^{B_2} \left| p_{Y|X}(y|w) \log p_{Y|X}(y|w) \right| \ud y \nonumber\\
& \le \int_{B_1}^{B_2} \left| \frac{1}{1+w^2} \right| \left| \exp \{ -\frac{y}{1+w^2} \} \log \frac{1}{1+w^2} \right| \ud y \nonumber \\
& \phantom{=} + \int_{B_1}^{B_2} \left| \frac{1}{1+w^2} \right| \left| \exp \left\{ -\frac{y}{1+w^2} \right\} \log \left( \exp \left\{ -\frac{y}{1+w^2} \right\} \right) \right| \ud y \nonumber \\
& \le \int_{B_1}^{B_2} \frac{1}{y^{2}} \ud y <\epsilon \nonumber
\end{align}
where again $B_0$ is chosen to ensure $1/B_0<\epsilon$.
\item
We finally have to address condition \ref{cond4}. Consider
\begin{align}\label{Ray}
\int_{0}^{+\infty} p_{Y|X}&(y|x) \log p_{Y|X}(y|x) \ud y - \int_{0}^{+\infty} p_{Y|X}(y|x) \log p_Y(y;p_X) \ud y \nonumber \\
&= \int_{0}^{+\infty} s\exp\{-ys\} \log \left( s\exp\{-ys\} \right) \ud y - \int_{0}^{+\infty} s\exp\{-ys\} \log p_Y(y;p_X) \ud y \nonumber\\
&= \log s - \frac{1}{\ln 2} - \int_{0}^{+\infty} s\exp\{-ys\} \log p_Y(y;p_X) \ud y.
\end{align}
The third term dependence\footnote{Dependence on variable $s$ is the same independently of the $s$ considered: It is thus possible to consider values for $s$ even outside the region dictated by the particular channel capacity problem we are considering.} on $s$ cannot be logarithmic since 
\begin{align}
\lim_{s \to +\infty} - \int_{0}^{+\infty} s\exp\{-ys\} \log p_Y(y;p_X) \ud y =0 \nonumber
\end{align}  
where exchange between integral and limit is licit since when $s \to +\infty$ it can be supposed greater than $1$, this ensuring the existence of an integrable upper bound of $|s\exp\{-ys\}\log p_Y(y;p_X)|$, much as previously done for integrability of $Q(y)\log q(y)$. Hence the difference between the first and third term of \eqref{Ray} cannot be constant on $\mathbb{S}$, this proving condition \ref{opt3} to hold.
\end{enumerate}

\section{Conclusion}\label{sec:conclusions}
This paper has proposed general conditions on the conditional output \ac{p.m.} under which real scalar channel models, with input peak power constraints, show to have capacity achieving \ac{p.m.}'s which are finitely discrete. These conditions represent a step towards a full understanding of the basic channel characteristics that determine the capacity achieving input \ac{p.m.} to be finitely discrete under peak power constraints. The here presented theory of peak power limited channels unifies under a same framework several channel models that were previously investigated using separated approaches, as shown by the provided examples.

Particular attention will be paid in the future to whether all of the supposed conditions are strictly necessary. Our feeling is that some of those conditions are not negotiable while other ones may not be as fundamental as they can appear to be. 

As last but not least consideration, we have matured the deep belief that only \emph{real} scalar peak power limited channels can have a finitely discrete capacity achieving input probability measure.

\appendices

%

\section{Boundedness and Continuity of the Marginal Information Density} \label{sec:PrelSteps}
The existence and boundedness of the upper and lower bounds on $p_{Y|X}(y|x)$, postulated in Section~\ref{sec:Channel} is sufficient to prove the existence and boundedness of $p_Y(y;p_X)$. In fact, we can write
\begin{align}
& q(y)= \int\limits_{\mathbb{S}} q(y) p_X(x) \ud x \le \int\limits_{\mathbb{S}} p_{Y|X}(y|x) p_X(x) \ud x 
\le \int\limits_{\mathbb{S}} Q(y) p_X(x) \ud x = Q(y) \nonumber
\end{align}
that is
\begin{align} \label{q_Q}
0 \le q(y) \le p_Y(y;p_X) \le Q(y) \le K, \;\;\; \forall y \in \mathbb{R} \; \mathrm{and} \; \forall \, p_X(x) \in \mathscr{P}.
\end{align}
An equally useful inequality, immediately descending from the previous one, is the following:
\begin{align} \label{log}
-\log{Q(y)} \le & -\log{p_Y(y;p_X)} \le -\log{q(y)}, \;\;\; 
\forall y \in \mathbb{R} \; \mathrm{and} \; \forall \, p_X(x) \in \mathscr{P}.
\end{align}
Moreover, consider the pair of functions $f(y)$ and $g(y)$, respectively nonnegative and positive, such that $g(y) \le K < +\infty$. The next inequality holds:
\begin{align} \label{f_g}
\left| f(y) \log g(y) \right| &\le -f(y) \log \frac{g(y)}{K} + f(y)|\log{K}|  \nonumber\\
& \le -f(y)\log g(y) + 2 f(y)|\log K|.
\end{align}
Besides
\begin{align}
G(y)=-Q(y)\log q(y) + 2 Q(y)|\log K| \nonumber
\end{align}
is integrable on $\mathbb{R}$. Proof for this is an immediate consequence of the conditions in Sec. \ref{sec:Channel}.

We now show that $h(x;p_X)$ and $i(x;p_X)$ are bounded $\forall x \in \mathbb{S}$ and $\forall \; p_X(x) \in \mathscr{P}$. In fact we have
\begin{align}
|h(x;p_X)| & =\left| \int\limits_{\mathbb{R}} p_{Y|X}(y|x) \log p_Y(y;p_X) \ud y \right| \le 
\int\limits_{\mathbb{R}} \left| p_{Y|X}(y|x) \log p_Y(y;p_X) \right| \ud y \nonumber \\
& \le \int\limits_{\mathbb{R}} p_{Y|X}(y|x) \big[-\log p_Y(y;p_X) + 2 |\log{K}|\big] \ud y 
\le \int\limits_{\mathbb{R}} Q(y) [-\log q(y) + 2 |\log{K}|] \ud y \nonumber\\
& = \int\limits_{\mathbb{R}} G(y) \ud y < +\infty \nonumber
\end{align}
having used \eqref{qQ}, \eqref{q_Q}, \eqref{log} and \eqref{f_g}. Moreover, we have
\begin{align}
|d(x)| & =\left| - \int\limits_{\mathbb{R}} p_{Y|X}(y|x) \log p_{Y|X}(y|x) \ud y \right| 
\le \int\limits_{\mathbb{R}} \left| p_{Y|X}(y|x) \log p_{Y|X}(y|x) \right| \ud y \nonumber \\
& \le \int\limits_{\mathbb{R}} p_{Y|X}(y|x) \left[ -\log p_{Y|X}(y|x) + 2 |\log{K}| \right] \ud y 
\le \int\limits_{\mathbb{R}} Q(y) \left[ -\log q(y) + 2 |\log{K}|\right] \ud y \nonumber\\
& = \int\limits_{\mathbb{R}} G(y) \ud y < +\infty \nonumber
\end{align}
where we again exploited \eqref{qQ}, \eqref{q_Q}, \eqref{log} and \eqref{f_g}. We may then conclude that $i(x;p_X)=h(x;p_X)-d(x)$ is bounded, as it is the difference between two quantities fulfilling the same finite boundedness property.

Continuity of $i(x;p_X)$ can be demonstrated in an almost identical way since, $\forall x \in \mathbb{S}$, it is possible to exchange the continuity limit with the integral in the definition of $i(\cdot)$, this being guaranteed by integrability of $G(y)$, and continuity of the integrand functions being an immediate evidence.

\section{Proof of Equation \eqref{Derivative}} \label{sec:Der}
The weak derivative can be developed as
\begin{align}
\phantom{=\;} & \lim_{\theta \to 0}  \frac{ I((1-\theta)F_1 + \theta F_2) - I(F_1) } {\theta} 
\nonumber \\
&= \lim_{\theta \to 0} \! \left\{ \frac{1}{\theta} \int_{\mathbb{R}} \int_{\mathbb{S}} \! p_{Y|X}(y|x) \! \log \! \frac {p_{Y|X}(y|x)} {\!p(y;(1\!-\!\theta)F_1 \!\!+\! \theta F_2)} \right. 
[(1\!-\!\theta)\ud F_1\!(x)\!\! + \!\! \theta \ud F_2\!(x)] \ud y \nonumber \\
& \phantom{=\;} - \; \left. \frac{1}{\theta} \int_{\mathbb{R}} \int_{\mathbb{S}} p_{Y|X}(y|x) \log \frac {p_{Y|X}(y|x)} {p_Y(y;F_1)} \ud F_1(x) \ud y \right\} \nonumber \\
& = \lim_{\theta \to 0} \left\{ \frac {1} {\theta} \int_{\mathbb{R}} \int_{\mathbb{S}} p_{Y|X}(y|x) \left[ 
- \log p_Y(y;F_\theta) 
+ \log p_Y(y;F_1) \right] \ud F_1(x) \ud y \right. \nonumber \\
& \phantom{=\;} + 
\int_{\mathbb{R}} \int_{\mathbb{S}} p_{Y|X}(y|x) \log \frac {p_{Y|X}(y|x)} { (1-\theta)p_Y(y;F_1) + \theta p_Y(y;F_2) } 
[\ud F_2(x)- \ud F_1(x)] \ud y \Big\} \nonumber \\ 
& \overset{\mathrm{(a)}}{=} \lim_{\theta \to 0} \left\{ \frac {1} {\theta} \int_{\mathbb{R}} \int_{\mathbb{S}} p_{Y|X}(y|x) \left[ \log p_Y(y;F_1) \right. 
- \log \left( (1-\theta)p_Y(y;F_1) + \theta p_Y(y;F_2) \right) \right] \ud F_1(x) \ud y \Big\} \nonumber \\
& \phantom{=\;}+ \int_{\mathbb{S}} i(x;F_1) \ud F_2(x) - I(F_1) \nonumber \\
& \overset{\mathrm{(b)}}{=} \lim_{\theta \to 0} \left\{ \frac {1} {\theta} \int_{\mathbb{R}} \int_{\mathbb{S}} p_{Y|X}(y|x) \Big[ \log p_Y(y;F_1) - \log p_Y(y;F_1) \right. \nonumber \\
& \phantom{=\;}- \left. \frac{1}{\ln 2} \theta \frac{- p_Y(y;F_1) + p_Y(y;F_2)}{p_Y(y;F_1)} \Big] \ud F_1(x) \ud y \right\} 
+ \int_{\mathbb{S}} i(x;F_1) \ud F_2(x) - I(F_1) \nonumber \\
& = \frac{1}{\ln 2} \int_{\mathbb{R}} \int_{\mathbb{S}} p_{Y|X}(y|x) \left( 1 - \frac{p_Y(y;F_2)}{p_Y(y;F_1)} \right) \ud F_1(x) \ud y + \int_{\mathbb{S}} i(x;F_1) \ud F_2(x) - I(F_1) \nonumber\\
& = \int_{\mathbb{S}} i(x;F_1) \ud F_2(x) - I(F_1) 
+ \frac{1}{\ln 2} \int_{\mathbb{R}} p_Y(y;F_1) \left( 1 - \frac{p_Y(y;F_2)}{p_Y(y;F_1)} \right) \ud y  \nonumber \\
& = \int_{\mathbb{S}} i(x;F_1) \ud F_2(x) - I(F_1) \nonumber
\end{align}
where the exchange between limit and integral in $\mathrm{(a)}$ follows from the Lebesgue Dominated Convergence Theorem. In fact, $\forall \theta \in [0,1]$ 
\begin{align}
f_\theta(y,x)p_{X_2}(x) &= p_{Y|X}(y|x)p_{X_2}(x) \log\frac{p_{Y|X}(y|x)}{p_Y(y;F_\theta)} \nonumber \\
&\le p_{X_2}(x)\left(\left|p_{Y|X}(y|x)\log{p_{Y|X}(y|x)}\right| + \left|p_{Y|X}(y|x)\log{p_{Y}(y;F_\theta)}\right|\right) \nonumber
\end{align}
which is integrable on $\mathbb{R}\times \mathbb{S}$,\footnote{Integration on $\mathbb{R}$ produces $p_{X_2}(x)i(x;F_\theta)$ that is integrable on $\mathbb{S}$ due to the boundedness of $i(\cdot)$.} by integrability of $G(y)$, and then also on $\mathbb{S}\times \mathbb{R}$ via Tonelli and Fubini Theorems and due to the fact that $f_\theta(y,x)$ converges, for $\theta \to 0$, to $f(y,x)=p_{Y|X}(y|x)\log\frac{p_{Y|X}(y|x)}{p_Y(y;F_1)}$.
Moreover, $\mathrm{(b)}$ follows from the first order McLaurin Series
\begin{align}
\log(a(1-x)+bx)\! = \!\log(a) \!+ \!\frac{x(-a+b)}{a} \frac{1}{\ln 2} +o(x). \nonumber
\end{align}

\section{An upper and lower bound for the Rayleigh fading conditional output \ac{p.m.}} \label{sec:Qandq}
In this appendix we rigorously prove inequality \eqref{qQ} to be satisfied in case the considered conditional output \ac{p.m.} and correspondent $Q(y)$ and $q(y)$ are the ones introduced in Section~\ref{RayChannel}. Concerning the upper bound, we have to show that there exist a parameter $\gamma$ such that
\begin{align}\label{upper}
\frac{1}{y^{1+\gamma}}>\frac{1}{1+x^2}\exp\left\{-\frac{y}{1+x^2}\right\}, \quad \forall x: 0\le |x|\le A 
\end{align}
is valid for $y>c(A^2+1)$, where $c>2$. The considered inequality can be reformulated as follows
\begin{align}
\frac{y}{(1+x^2)(1+\gamma)}+\frac{\ln(1+x^2)}{1+\gamma}>\ln y. \nonumber
\end{align}
To guarantee the inequality to be fulfilled even in the worst case, the left hand side ($x$ is confined in it) can be studied, for each fixed $y$, to find out that
$\sqrt{y-1}$ is its minimum in $x$, provided $y\geq 1$. Moreover, if $y\geq A^2+1$ the minimum becomes $x=A$, since $x$ is bounded and $\sqrt{y-1}$ is unreachable in this case. The minimum expression for $y\geq A^2+1$ is
\begin{align}
\frac{y}{(1+A^2)(1+\gamma)}+\frac{\ln(1+A^2)}{1+\gamma} \nonumber
\end{align}
which has constant derivative in $y$ equalling the derivative of $\ln y$ at $y=(1+A^2)(1+\gamma)$. Now, consider $y_1=c(1+A^2)$: If constants $c$ and $\gamma$ are chosen such that $c>2$ and $\gamma <1$, then $y_1>(1+A^2)(1+\gamma)$ and
\begin{align}
\left.\frac{\ud}{\ud y} \left[\frac{y}{(1+A^2)(1+\gamma)}+\frac{\ln(1+A^2)}{1+\gamma}\right]\right|_{y_1} > \left.\frac{\ud}{\ud y} \ln y \right|_{y_1}. \nonumber
\end{align}
This ensures the derivative of $\frac{y}{(1+A^2)(1+\gamma)}+\frac{\ln(1+A^2)}{1+\gamma}$ to be greater than the one of $\ln y$, which is decreasing, for $y>y_1$. If, finally, it is possible to derive a condition on $\gamma$ to provide that
\begin{align}\label{minln}
\left.\left(\frac{y}{(1+A^2)(1+\gamma)}+\frac{\ln(1+A^2)}{1+\gamma}\right)\right|_{y_1}> \left.\ln y \right|_{y_1}
\end{align}
the original assertion \eqref{upper} would be satisfied. This is indeed possible since \eqref{minln} becomes
\begin{align}
\frac{c}{1+\gamma}+\frac{\ln(1+A^2)}{1+\gamma}>\ln(1+A^2)+\ln c \nonumber
\end{align}
which is satisfied for $\gamma<\frac{c-\ln c}{\ln\left[c(1+A^2)\right]}$. Any choice of $\gamma$ such that 
\begin{align}
\gamma<\min\left\{1,\frac{c-\ln c}{\ln\left[c(1+A^2)\right]}\right\} \nonumber
\end{align}
would fulfil the scope. Consequently the definition
\begin{align}
Q(y)=\begin{cases}
1, & \quad 0 \le y \le c(A^2+1) \\
\frac{1}{y^{1+\gamma}}, &\quad y>c(A^2+1) \nonumber
\end{cases}
\end{align}
is well posed since it guarantees the right hand side of inequality \eqref{qQ} to be respected.

Concerning the lower bound $q(y)$, we have to prove that it coincides with the output \ac{p.m.} conditioned by the maximum input up to $y_2=\frac{(1+A^2)\ln(1+A^2)}{A^2}$ and that it coincides with the output \ac{p.m.} conditioned by the minimum input after that same $y_2$. To do that, consider the intersection between $\frac{1}{1+x^2}\exp\left\{-\frac{y}{1+x^2}\right\}$ and $\exp\{-y\}$ which is given by
\begin{align}
y(x)=\frac{(1+x^2)\ln(1+x^2)}{x^2}. \nonumber
\end{align}
This intersection is non decreasing in $x$ for $0 \le x \le A$ (only positive values are admissible for $x$, deriving from normalization in \cite{FayTroSha:01}) since
\begin{align}
\frac{\ud y(x)}{\ud x} = \frac{2}{x^3}\left[ x^2-\ln(1+x^2) \right]\geq 0 \nonumber
\end{align}
this meaning that it is maximum for $x=A$. This duly proves that the output \ac{p.m.} conditioned by the maximum input lies under all the other conditional output \ac{p.m.} up to its intersection with $\exp\{-y\}$ in $y_2=\frac{(1+A^2)\ln(1+A^2)}{A^2}$, while afterwards the same role is taken by $\exp\{-y\}$. This finally proves that
\begin{align}
q(y)=\begin{cases}
\frac{1}{1+A^2}\exp \left\{-\frac{y}{1+A^2}\right\}, & \quad 0 \le y \le \frac{(1+A^2)\ln(1+A^2)}{A^2} \\
\exp \{-y\}, & \quad y>\frac{(1+A^2)\ln(1+A^2)}{A^2} \nonumber
\end{cases}
\end{align}
is also well posed, fulfilling the left hand side of \eqref{qQ}.

\section{Useful Theorems} \label{sec:Theorems}
This appendix provides a collection of theorem statements (along with the appropriate references) that are used throughout this paper. 

\medskip
\begin{theorem}[Helly's Weak Compactness Theorem \cite{Loe:B77}]
Every sequence of \ac{c.d.f.}'s is weakly compact.\footnote{Recall that a set is said to be \emph{compact}, in the sense of a type of convergence, if every infinite sequence in the set contains a subsequence which is convergent in that same sense \cite{Loe:B77}.}
\end{theorem}

\medskip
\begin{theorem}[Helly-Bray Theorem \cite{Loe:B77}]
If $g$ is continuous and bounded on $\mathbb{R}^n$, then $F_n \xrightarrow[n]{c} F$ up to additive constants implies $\int g \, \ud F_n \to \int g \, \ud F$.
\end{theorem} 
This theorem is formulated in terms of complete convergence, but complete convergence is equivalent to L\`evy convergence in $\mathscr{F}$.
\begin{theorem}[Differentiation Lemma \cite{Lan:B99}]
Let $I$ be an interval of real numbers, eventually infinite, and $U$ be an open set of complex numbers. Let $f=f(t,z)$ be a continuous function on $I\times U$. Assume 1) for each compact subset $\mathbb{K}$ of $U$ the integral $\int\limits_{I} f(t,z)\ud t$ is uniformly convergent for $z \in \mathbb{K}$, 2) for each $t$, the function $z\mapsto f(t,z)$ is analytic, then the integral function $F(z)=\int\limits_{I} f(t,z) \ud t$ is analytic on $U$.
\end{theorem}

\medskip
\begin{theorem}[\cite{Lue:B69}, Proposition 1, Chapter 7.8]\label{ConFunc}
Let f be a concave functional defined on a convex subset C
of a normed space. Let $\mu = \sup_{x\in C} f(x)$. Then 
\begin{enumerate}
\item The subset $\Omega$ of C where $f(x) = \mu$ is convex. 
\item If $x_0$ is a local maximum of $f(\cdot)$, then $f(x_0) = \mu$ and, hence $x_0$ is a global maximum.
\end{enumerate}
\end{theorem}



\section*{Acknowledgment}

The authors would like to thank Prof. Marco Chiani, Prof. Massimo Cicognani, Andrea Mariani, Simone Moretti and Matteo Mazzotti for useful comments and discussions.

\ifCLASSOPTIONcaptionsoff
  \newpage
\fi



\bibliographystyle{IEEEtran}
\bibliography{IEEEabrv,StringDefinitions,BiblioVZ,BiblioDD,RFID,MyBooks,ITVince,BiblioOFDM}

\begin{thebibliography}{10}
\providecommand{\url}[1]{#1}
\csname url@samestyle\endcsname
\providecommand{\newblock}{\relax}
\providecommand{\bibinfo}[2]{#2}
\providecommand{\BIBentrySTDinterwordspacing}{\spaceskip=0pt\relax}
\providecommand{\BIBentryALTinterwordstretchfactor}{4}
\providecommand{\BIBentryALTinterwordspacing}{\spaceskip=\fontdimen2\font plus
\BIBentryALTinterwordstretchfactor\fontdimen3\font minus
  \fontdimen4\font\relax}
\providecommand{\BIBforeignlanguage}[2]{{%
\expandafter\ifx\csname l@#1\endcsname\relax
\typeout{** WARNING: IEEEtran.bst: No hyphenation pattern has been}%
\typeout{** loaded for the language `#1'. Using the pattern for}%
\typeout{** the default language instead.}%
\else
\language=\csname l@#1\endcsname
\fi
#2}}
\providecommand{\BIBdecl}{\relax}
\BIBdecl

\bibitem{Smi:69}
J.~G. Smith, ``On the information capacity of peak and average power
  constrained {G}aussian channels,'' \emph{Ph.D. dissertation, University of
  California, Berkeley}, 1969.

\bibitem{Smi:71}
------, ``The information capacity of amplitude- and variance-constrained
  scalar {G}aussian channels,'' \emph{Information and Control}, vol.~18, 1971.

\bibitem{Sha:48}
C.~E. Shannon, ``A {M}athematical {T}heory of {C}ommunication,'' \emph{Bell
  System Technical Journal}, Jul./Oct. 1948.

\bibitem{ShaBar:95}
S.~Shamai and I.~Bar-David, ``The {C}apacity of {A}verage and
  {P}eak-{P}ower-{L}imited {Q}uadrature {G}aussian {C}hannels,'' \emph{{IEEE}
  Trans. Inf. Theory}, vol.~41, no.~4, pp. 1060--1071, Jul. 1995.

\bibitem{FayTroSha:01}
I.~Abou-Faycal, M.~Trott, and S.~Shamai, ``The {C}apacity of {D}iscrete-{T}ime
  {M}emoryless {R}ayleigh-{F}ading {C}hannels,'' \emph{{IEEE} Trans. Inf.
  Theory}, vol.~47, no.~4, pp. 1290 --1301, May 2001.

\bibitem{Tch:04}
A.~Tchamkerten, ``On the {D}iscreteness of {C}apacity-{A}chieving
  {D}istributions,'' \emph{{IEEE} Trans. Inf. Theory}, vol.~50, no.~11, pp.
  2773--2778, Nov. 2004.

\bibitem{ChaHraKsc:05}
T.~H. Chan, S.~Hranilovic, and F.~R. Kschischang, ``{C}apacity-{A}chieving
  {P}robability {M}easure for {C}onditionally {G}aussian {C}hannels {W}ith
  {B}ounded {I}nputs,'' \emph{{IEEE} Trans. Inf. Theory}, vol.~51, no.~6, pp.
  2073--2088, Jun. 2005.

\bibitem{FeiMat:07}
A.~Feiten and R.~Mathar, ``Capacity-{A}chieving {D}iscrete {S}ignaling over
  {A}dditive {N}oise {C}hannels,'' in \emph{Proc. IEEE Int. Conf. on Commun.},
  Jun. 2007, pp. 5401--5405.

\bibitem{LeiGeiWit:12}
E.~Leitinger, B.~C. Geiger, and K.~Witrisal, ``{C}apacity and
  {C}apacity-{A}chieving {I}nput {D}istribution of the {E}nergy {D}etector,''
  in \emph{IEEE Int. Conf. on Utra-Wideband}, Sep. 2012, pp. 57--61.

\bibitem{CovTho:B91}
T.~A. Cover and J.~A. Thomas, \emph{Elements of Information Theory},
  1st~ed.\hskip 1em plus 0.5em minus 0.4em\relax New York, NY, 10158: John
  Wiley \& Sons, Inc., 1991.

\bibitem{Lan:B99}
S.~Lang, \emph{Complex Analysis}.\hskip 1em plus 0.5em minus 0.4em\relax New
  York: Springer-Verlag, 1999.

\bibitem{Lue:B69}
D.~Luenberger, \emph{Optimization by Vector Space Methods}.\hskip 1em plus
  0.5em minus 0.4em\relax New York: John Wiley \& Sons, 1969.

\bibitem{Mor:B68}
P.~Moran, \emph{An Introduction to Probability Theory}.\hskip 1em plus 0.5em
  minus 0.4em\relax Oxford: Clarendon Press, 1968.

\bibitem{Loe:B77}
M.~Lo\`{e}ve, \emph{Probability Theory I}, 4th~ed.\hskip 1em plus 0.5em minus
  0.4em\relax New York, NY: Springer-Verlag, 1977.

\end{thebibliography}
\end{document}